\documentclass[a4paper,10pt]{article}
\usepackage{amsmath, amsthm, amsfonts, amssymb, bbm}
\usepackage[mathscr]{eucal}
\usepackage[latin1]{inputenc}
\usepackage{axodraw}

\theoremstyle{plain}
\newtheorem{proposition}{Proposition}[section]

\theoremstyle{definition}

\theoremstyle{remark}
\newtheorem{remark}[proposition]{Remark}

\newcommand{\rhs}{r.h.s.\ }
\newcommand{\lhs}{l.h.s.\ }
\newcommand{\wrt}{w.r.t.\ }

\newcommand{\ud}{\mathrm{d}}
\newcommand{\del}{\partial}
\newcommand{\bra}[1]{\langle #1 \rvert}
\newcommand{\ket}[1]{\lvert #1 \rangle}
\newcommand{\braket}[2]{\langle #1 \vert #2 \rangle}
\newcommand{\ketbra}[2]{\lvert #1 \rangle \langle #2 \rvert}

\newcommand{\skal}[2]{\langle #1 , #2 \rangle}

\DeclareMathOperator{\Tr}{Tr}

\DeclareMathOperator{\sign}{sign}

\newcommand{\R}{\mathbb{R}}

\newcommand{\F}{\mathcal{F}}

\newcommand{\betrag}[1]{\lvert #1 \rvert}

\numberwithin{equation}{section}

\begin{document}

\begin{flushright}
ITP-UH-08/10
\end{flushright}
\vspace{1cm}
\hfill
\begin{center}
{\Large \bf Divergences in Quantum Field Theory on the}
\vspace{0.1cm}
\hfill
{\Large \bf Noncommutative Two--Dimensional Minkowski}
\vspace{0.16cm}
\hfill
{\Large \bf Space with Grosse--Wulkenhaar Potential}

\hfill \\
{\large Jochen Zahn \\ Institut f\"ur Theoretische Physik, Leibniz Universit\"at Hannover \\ Appelstra{\ss}e 2, 30167 Hannover, Germany \\ jochen.zahn@itp.uni-hannover.de \\
\hfill \\
May 4, 2010 \\}
\hfill \\
\end{center}


\begin{abstract}
Quantum field theory on the noncommutative two-dimension\-al Minkowski space with Grosse-Wulkenhaar potential is discussed in two ways: In terms of a continuous set of generalised eigenfunctions of the wave operator, and directly in position space. In both settings, we find a new type of divergence in planar graphs. It is present at and above the self-dual point. This new kind of divergence might make the construction of a Minkowski space version of the Grosse-Wulkenhaar model impossible.
\end{abstract}

\section{Introduction}

The extensive study of noncommutative quantum field theories (NCQFT) that started about fifteen years ago has undergone several twists since then (for motivations and an overview, we refer to the reviews \cite{Review}). Right from the start, different approaches were followed. While Filk derived modified Feynman rules for the Euclidean case \cite{Filk}, Doplicher, Fredenhagen and Roberts started on the noncommutative Minkowski space and proposed a Hamiltonian setting for the quantisation of field theories \cite{DFR}. In the context of Filk's Feynman rules, the phenomenon of UV/IR mixing \cite{Minwalla} was found. For a couple of years, this was an obstacle for a systematic treatment of renormalisation.

The Hamiltonian approach and the modified Feynman rules are equivalent when time commutes with all spatial coordinates. However, when this is not the case, the naive application of Filk's Feynman rules to field theories on the noncommutative Minkowski space leads to a violation of unitarity, in the sense that the cutting rules no longer hold \cite{GomisMehen}. As Bahns et al. \cite{BDFP02} pointed out, this violation of unitarity is the consequence of an inappropriate definition of time-ordering and not present in the Hamiltonian setting \cite{Hamiltonian}. As another alternative, they proposed the Yang-Feldman approach \cite{YF}. Then, the UV/IR mixing manifests itself as a distortion of the dispersion relations \cite{Quasiplanar, NCDispRel}.

In recent years, the study of models with an added harmonic oscillator potential became popular. Such a modification was first proposed by Langmann and Szabo \cite{LangmannSzabo}, who showed that in such a way the action of the complex $\phi^4$ theory can be made invariant under Fourier transformation. This occurs for a particular value of the harmonic oscillator frequency, the so-called self-dual point\footnote{In this work, this term does only refer to the frequency of the harmonic oscillator potential. The models discussed here are in general not self-dual in the sense of \cite{LangmannSzabo}.}.  For the case where the harmonic oscillator potential is replaced by a constant magnetic field (of the same frequency), the model is solvable (but trivial)~\cite{LSZ}. Further evidence for the need for an added harmonic oscillator potential came from the success of the Grosse-Wulkenhaar model. They showed that with such a modification, the noncommutative $\phi^4$ model is renormalisable not only in two~\cite{GrosseWulkenhaar2d}, but also in four spacetime dimensions~\cite{GrosseWulkenhaar4d}. Even better, the model is asymptotically safe (but not free!), since the $\beta$-function is bounded~\cite{Ghost,DisertoriEtAl}.

In this approach, one uses (Weyl symbols of) ket-bras of harmonic oscillator eigenfunctions to transform the model to matrix form. Then, the interaction term takes a particularly simple form. It is precisely at the self-dual point that also the propagator becomes simple (diagonal). It turns out that in four spacetime dimensions\footnote{In two spacetime dimensions, the quadratic potential is only needed in an intermediate step. In the limit where the cutoff in the matrix base is removed, its frequency vanishes \cite{GrosseWulkenhaar2d}.}, the self-dual point is a fixed point of the theory 
\cite{Ghost,DisertoriEtAl}.

While the study of Euclidean models with an added harmonic oscillator potential was a spectacular success, very little is known about the corresponding Minkowski space versions\footnote{In \cite{FischerSzabo}, first steps in this direction were taken. The relation of their approach to the ones discussed here is clarified in Appendix~\ref{app:FischerSzabo}.}. As we show here, the self-dual point is a special point also on Minkowski space, but in an unexpected way: It is at this point that a strange kind of divergence appears in planar graphs. At first sight, this is paradoxical: Usually the planar graphs are as in the commutative case. But in two dimensions all scalar field theories are superrenormalisable. Since the degree of the singularity of the propagator only depends on the kinetic term\footnote{This is the case for the retarded propagator \cite{Friedlander}. That the same is also true for the Feynman propagator and the two-point function should rather be viewed as a condition for suitable states. Basically this is the Hadamard, or in modern terms, the microlocal spectrum condition for QFT on curved spacetimes \cite{Microlocal}.}, this is also true if a non-constant potential is added. Hence, there is no divergence in the commutative case and we would not expect to find one in the planar sector.
However, the fact that planar graphs are exactly as in the commutative case relies on the cancellation of the twisting factors. This, in turn, happens due to momentum conservation. But momentum is not conserved because of the quadratic potential. Thus, there is no reason to expect that the planar part is exactly as in the commutative case.

We will find these divergences in two different ways. In the first approach, we restrict ourselves to the self-dual point and use a continuous set of generalised eigenfunctions of the wave operator as a basis for quantisation. As a first step, we assume naive Feynman rules and compute the fish graph in the $\phi^3$ model. While there is some ambiguity stemming from different possible choices for the Feynman propagator, one generically finds a peculiar type of divergence in planar graphs. It is no UV divergence in that it does not stem from a divergent loop integral. Instead, it comes from the kinematical factors at the vertices. If the divergence is formally removed, one also finds a violation of unitarity, in the sense that the cutting rules are not fulfilled. This is not surprising given that we postulated Feynman rules without caring for correct time-ordering. This problem can be cured by quantising \`a la Yang-Feldman. However, we will argue that this does not remove the divergences.

The second approach that leads to the same conclusion is a treatment of the model in position space. In that case, one is not restricted to the self-dual point. In order to circumvent the ambiguities connected to a choice of the Feynman propagator, we start by explicitly constructing the retarded propagator. It turns out to be increasing like a Gaussian in some directions, so we interpret it as a distribution on a Gelfand-Shilov space. We show that for distributions on that space, the planar $\star$-product at different points can not be defined via duality if one is at or above the self-dual point. Again, this shows that the problem is not a UV divergence, since it occurs before taking the limit of coinciding points. We also show that when the planar $\star$-product is calculated in a formal sense, one finds a geometric series that diverges at and above the self-dual point.

The appearance of this new kind of divergence is an interesting phenomenon, that, to our opinion, deserves more detailed studies. In particular, it should be checked whether it also occurs in four spacetime dimensions, since there the self-dual point is a fixed point of the theory \cite{Ghost,DisertoriEtAl}, at least in the Euclidean case.

This paper is organised as follows: In the next section, we fix some notation. In Section~\ref{sec:Eigenfunctions} we introduce the continuous basis of generalised eigenfunctions. This is used in Section~\ref{sec:Causality} to postulate naive Feynman rules and discuss the fish graph of the $\phi^3$ model in this setting. We find the above mentioned violation of unitarity and the new type of divergence in the planar sector. In order to better understand these, we construct the retarded propagator and discuss some of its properties in Section~\ref{sec:Delta_ret}. It is then used to further
analyse the planar divergences in Section~\ref{sec:PlanarDiv}. We conclude with a summary and an outlook.
In Appendix~\ref{app:FischerSzabo}, we discuss the relation to the matrix model introduced in \cite{FischerSzabo}. Appendix~\ref{sec:Fourier} contains the calculation of the retarded propagator in momentum space.

\section{Setup}

We start by fixing some notation: For the commutation relations, we write
\begin{equation}
\label{eq:Commutator}
[x^\mu, x^\nu] = i \theta^{\mu \nu} = i \lambda_{\text{nc}}^2 \epsilon^{\mu \nu}
\end{equation}
with
\[
\epsilon = \begin{pmatrix} 0 & 1 \\ -1 & 0 \end{pmatrix}.
\]
Here we introduced a length scale $\lambda_{\text{nc}}$.
The product of functions of these noncommuting coordinates can now be defined via the $\star$-product,
\begin{equation}
\label{eq:StarProduct}
(f \star g)(x) = f e^{\frac{i}{2} \overleftarrow{\del_\mu} \theta^{\mu \nu} \overrightarrow{\del_\nu}} g(x),
\end{equation}
or by the twisted convolution
\begin{equation}
\label{eq:TwistedConvolution}
(f \star g) \hat{\ } (\tilde k) = (2\pi)^{-1} \int \ud^2 k \ \hat{f}(k) \hat{g}(\tilde k - k) e^{-\frac{i}{2} k_\mu \theta^{\mu \nu} \tilde k_\nu},
\end{equation}
where the hat denotes Fourier transformation. For analytic functions, these are equivalent. When in doubt, we use (\ref{eq:TwistedConvolution}).

The Grosse-Wulkenhaar potential can now be introduced in the following way: We define \cite{LangmannSzabo,GrosseWulkenhaar2d} 
\begin{align*}
D^\pm_\mu = - i \del_\mu \mp 2 \Omega {\theta^{-1}}_{\mu \nu} x^\nu = - i \del_\mu \pm 2 \lambda^{-2} \epsilon_{\mu \nu} x^\nu.
\end{align*}
Here we defined
\begin{equation}
\label{eq:lambda}
 \lambda = \Omega^{-\frac{1}{2}} \lambda_{\text{nc}}.
\end{equation}
The choice $\Omega = 1$, i.e., $\lambda = \lambda_{\text{nc}}$, corresponds to the self-dual point. Obviously,
\[
[D^\pm_\mu, D^\pm_\nu] = \pm 4 i \lambda^{-2} \epsilon_{\mu \nu}, \quad [D^\pm_\mu, D^\mp_\nu] = 0
\]
and
\[
D_\mu^\pm D^{\pm \mu} = - \del_\mu \del^\mu \mp 4 i \lambda^{-2} \epsilon_{\mu \nu} x^\nu \del^\mu - 4 \lambda^{-4} x_\mu x^\mu.
\]
The wave equation for a scalar field $\phi$ of mass $\mu$ in the quadratic potential is then given by
\begin{equation}
\label{eq:WaveOp}
\left( -\tfrac{1}{2} \left(D^+_\mu D^{+\mu} + D^-_\mu D^{-\mu}\right) + \mu^2 \right) \phi = \left( \del_\mu \del^\mu + 4 \lambda^{-4} x_\mu x^\mu + \mu^2 \right) \phi.
\end{equation}
This is starting point for a discussion of the model in position space. The reader who is interested in this approach may thus directly jump to Section~\ref{sec:Delta_ret}. The next two sections are devoted to the study of the model in terms of generalised eigenfunctions of the above wave operator for the case $\lambda = \lambda_{\text{nc}}$.

\section{The eigenfunctions}
\label{sec:Eigenfunctions}
We want to study the eigenfunctions of the wave operator \eqref{eq:WaveOp}. For this, we restrict ourselves to the self-dual point, i.e., to the case $\lambda = \lambda_{\text{nc}}$. The importance of this restriction is that then $D^+_\mu D^{+\mu}$ and $D^-_\mu D^{-\mu}$ can be represented as $\star$-multiplication from left, respectively right. In order to see this, we use the form (\ref{eq:StarProduct}) of the $\star$-product. For $H = - \frac{2}{\lambda^2} x_\mu x^\mu$, one obtains~\cite{FischerSzabo}, using $\epsilon^T \eta \epsilon = - \eta$,
\begin{align*}
H \star f & = - \tfrac{2}{\lambda^2} x_\mu x^\mu \left( 1 + \tfrac{i}{2} \lambda^2 \overleftarrow{\del_\nu} \epsilon^{\nu \lambda} \overrightarrow{\del_\lambda} - \tfrac{\lambda^4}{8} \overleftarrow{\del_\nu} \overleftarrow{\del_\rho} \epsilon^{\nu \lambda} \epsilon^{\rho \sigma} \overrightarrow{\del_\lambda} \overrightarrow{\del_\sigma} \right) f \\
& = - \tfrac{2}{\lambda^2} \left( x_\mu x^\mu + i \lambda^2 \epsilon_{\mu \nu} x^{\mu} \del^\nu + \tfrac{1}{4} \lambda^4 \del_\mu \del^\mu \right) f \\
& = \tfrac{\lambda^2}{2} D_\mu^- D^{- \mu} f.
\end{align*}
Analogously, one finds
\begin{equation*}
f \star H = \tfrac{\lambda^2}{2} D_\mu^+ D^{+ \mu} f.
\end{equation*}
Thus, if we find a complete set of orthonormal generalised eigenvectors $|ks\rangle$ of the Wigner transform of $H$, with eigenvalues $k$ and a degeneracy index $s$, then we have
\begin{equation}
\label{eq:chiWaveEq}
\left( - \tfrac{1}{2} \left(D^+_\mu D^{+\mu} + D^-_\mu D^{-\mu}\right) + \mu^2 \right) \chi_{kl}^{st} = \left( - \lambda^{-2}(k + l) + \mu^2 \right) \chi_{kl}^{st},
\end{equation}
where $\chi_{kl}^{st}$ is the Weyl symbol of the ket-bra operator $\ketbra{ks}{lt}$.
Furthermore, in this basis, the $\star$-product takes the form
\begin{equation}
\label{eq:chiStar}
\chi_{kl}^{st} \star \chi_{k'l'}^{s't'} = \delta(l-k') \delta^{ts'} \chi_{kl'}^{st'},
\end{equation}
and because of the cyclicity of the integral we have
\begin{equation}
\label{eq:chiInt}
\int \ud^2x \ \chi^{st}_{kl} = \delta(k-l) \delta^{st}.
\end{equation}
Indeed, a basis with the required properties exists. As shown in \cite{ChruscinskiII,FischerSzabo} and below, the spectrum of $H$ is the entire real line, with a two-fold degeneracy. The eigenvalues $k, l$ will be called the generalised momenta in the following.

In order to find the eigenfunctions of $H$, we implement the commutation relations \eqref{eq:Commutator} by choosing (we recall that here $\lambda_{\text{nc}} = \lambda$)
\begin{equation}
\label{eq:xRep}
x_0 = \lambda q, \quad x_1 = \lambda p,
\end{equation}
where $q$ and $p = - i \del_q$ are the position and momentum operators on $L^2(\R)$.
The Hamiltonian $H$ thus becomes
\[
H = - \tfrac{2}{\lambda^2} x_\mu x^\mu = 2(p^2 - q^2) = - \left( (q+p)(q-p) + (q-p)(q+p) \right).
\]
We write this as
\[
H = - 2(U V + V U),
\]
with
\begin{equation}
\label{eq:UV}
U = \tfrac{1}{\sqrt{2}} (q-p); \quad V = \tfrac{1}{\sqrt{2}} (q+p).
\end{equation}
We have
\[
[U,V] = i.
\]
Choosing the canonical representation for $U$ and $V$, we thus have to solve the eigenvalue equation
\[
2 i (u \del_u + \del_u u) \psi_k(u) = k \psi_k(u),
\]
or
\begin{equation}
\label{eq:Eigenfct}
u \del_u \psi_k(u) = \left( - i \tfrac{k}{4} - \tfrac{1}{2} \right) \psi_k(u)
\end{equation}
Generalised eigenfunctions that solve this are given by \cite{BolliniOxman}
\[
\psi_k^\pm(u) = \tfrac{1}{2 \sqrt{2\pi}} u^{- i \frac{k}{4} - \frac{1}{2}}_{\pm} = \left\{ \begin{aligned} \tfrac{1}{2 \sqrt{2\pi}} \betrag{u}^{- i \frac{k}{4} - \frac{1}{2}} & \text{ for } u \gtrless 0 \\ 0 & \text{ otherwise }  \end{aligned} \right..
\]
It is straightforward to prove the orthonormality relations
\begin{align*}
\braket{\psi_k^s}{\psi_{l}^t} & = \delta^{st} \delta(k-l) \\
\sum_s \int \ud k \ \bar \psi_k^s(u) \psi_{k}^s(u') & = \delta(u-u')
\end{align*}
There is a similar basis, obtained from $\psi_k^\pm$ by Fourier transformation, which is given by \cite{ChruscinskiI}
\[
\xi_k^\pm(u) = \tfrac{\pm i}{4\pi} e^{\mp \frac{i\pi}{2}(- i \frac{k}{4} - \frac{1}{2})} \Gamma(-i\tfrac{k}{4} + \tfrac{1}{2}) (u \pm i \epsilon)^{i\frac{k}{4} - \frac{1}{2}}.
\]

The change back to the $p, q$ representation is achieved by the unitary transformation
\begin{subequations}
\label{eq:chi_eta}
\begin{align}
\chi^k_\pm(q) & = (2\pi^2)^{-\frac{1}{4}} \int \ud u \ e^{i \left(\frac{q^2}{2} - \sqrt{2} qu + \frac{u^2}{2} \right) } \psi_k^\pm(u), \\
\eta^k_\pm(q) & = (2\pi^2)^{-\frac{1}{4}} \int \ud u \ e^{i \left(\frac{q^2}{2} - \sqrt{2} qu + \frac{u^2}{2} \right) } \xi_k^\pm(u).
\end{align}
\end{subequations}
As shown in \cite{ChruscinskiII}, the results are the parabolic cylinder functions that were used in \cite{FischerSzabo} and denoted by the same symbols (their convention is related to the one used here by $\mathcal{E} = \lambda^{-2} k/4$). However, we note that $U$ and $V$ are multiples of the light cone coordinates, which are very convenient in two dimensions. Defining, cf. (\ref{eq:xRep}) and (\ref{eq:UV}),
\begin{equation}
\label{eq:uv}
u = \tfrac{1}{\sqrt{2} \lambda} (x_0 - x_1), \quad v = \tfrac{1}{\sqrt{2} \lambda} (x_0 + x_1),
\end{equation}
we obtain the Weyl symbol of the ket-bras $\ket{\psi_k^s} \bra{\psi_l^t}$ in these coordinates as 
\[
\chi_{kl}^{st}(u,v) = \int \ud p \ e^{ipv} \braket{u-p/2}{\psi_k^s} \braket{\psi_l^t}{u+p/2}.
\]
We compute this explicitly for the $++$ component.
\[
\chi_{kl}^{++}(u, v) = \frac{1}{8 \pi} \int \ud p \ e^{i p v} (u - p/2)_+^{-i\frac{k}{4}-\frac{1}{2}} (u + p/2)_+^{i\frac{l}{4}-\frac{1}{2}}
\]
This vanishes for $u\leq 0$. For $u>0$, we obtain, using \cite[(13.2.1)]{Abramowitz},
\begin{align*}
\chi_{kl}^{++}(u, v) & = \frac{1}{8 \pi} \int_{-2u}^{2u} \ud p \ e^{i p v} (u - p/2)^{-i\frac{k}{4}-\frac{1}{2}} (u + p/2)^{i\frac{l}{4}-\frac{1}{2}} \\
& = \frac{(2u)^{i \frac{l-k}{4}}}{4 \pi} \int_{-1/2}^{1/2} \ud p \ \ e^{4 i p u v} (1/2 - p)^{-i\frac{k}{4}-\frac{1}{2}} (1/2 + p)^{i\frac{l}{4}-\frac{1}{2}} \\
& = \frac{(2u)^{i \frac{l-k}{4}}}{4 \pi} e^{-2iuv} \int_{0}^1 \ud p \ \ e^{4 i p u v} (1 - p)^{-i\frac{k}{4}-\frac{1}{2}} p^{i\frac{l}{4}-\frac{1}{2}} \\
& = \frac{(2u)^{i \frac{l-k}{4}}}{4 \pi} e^{-2iuv} \frac{\Gamma(i\frac{l}{4}+\frac{1}{2}) \Gamma(-i\frac{k}{4}+\frac{1}{2}) M(i\frac{l}{4}+\frac{1}{2},  i\frac{l-k}{4}+1, 4 i u v)}{\Gamma(i\frac{l-k}{4}+1)}
\end{align*}
where $M$ is Kummer's confluent hypergeometric function of the first kind. For the $--$ component, one finds a similar expression, and also the $+-$ and $-+$ components can be expressed in terms of special functions, in this case Kummer's confluent hypergeometric function of the second kind. However, for the present discussion, the explicit form of the eigenfunctions in position space is not relevant.

\section{Quantisation in terms of the eigenfunctions}
\label{sec:Causality}

We now want to discuss field theory at the self-dual point in terms of the continuous set of generalised eigenfunctions. Expanding fields in terms of $\chi_{kl}^{st}$,
\[
\phi = \sum_{st} \int \ud k \ud l \ \phi_{kl}^{st} \chi_{lk}^{ts} 
\]
where
\[
\phi_{kl}^{st} = \int \ud^2 x \ \chi_{kl}^{st} \star \phi, 
\]
we may write the wave operator $W$ in matrix notation as
\[
(W \phi)_{kl}^{st} = W^{st \ t's'}_{kl \ l'k'} \phi_{k'l'}^{s't'}
\]
with, cf. \eqref{eq:chiWaveEq},
\[
W^{st \ t's'}_{kl \ l'k'} = \left( - \lambda^{-2}(k + l) + \mu^2 \right) \delta(k-k') \delta(l-l') \delta^{ss'} \delta^{tt'}.
\]
This can easily be inverted to yield a propagator
\begin{equation}
\label{eq:Propagator}
\Delta^{st \ t's'}_{kl \ l'k'} = \frac{-\lambda^2}{ k + l - \lambda^2 \mu^2 + i \epsilon \sigma_{st}(k, l)} \delta(k-k') \delta(l-l') \delta^{ss'} \delta^{tt'}.
\end{equation}
Here $\sigma_{st}(k, l)$ is a sign function which can be chosen such as to achieve the required causality properties for the propagator.
Of course, the choice of this sign function affects the loop integrals we want to calculate later on. However, as we will see, some properties of the loop integrals are generic in that they do not depend on this choice. In particular, we find a divergence that is present even before evaluating the loop integral, and thus independent of the sign function.
Hence, we will not invest too much care into a rigorous discussion of the possible choices for the sign function.

For the graphical statement of naive Feyman rules, we use a double line notation, similar to \cite{GrosseWulkenhaar2d}. The two lines can be interpreted as the bra and the ket of the eigenfunctions. According to \eqref{eq:Propagator}, the propagator is given by
\vspace{-10pt}
\begin{center}
\begin{picture}(40,20)(0,10)
\ArrowLine(5,8)(35,8)
\ArrowLine(35,12)(5,12)
\Text(5,6)[t]{$ks$}
\Text(5,14)[b]{$lt$}
\Text(35,6)[t]{$k's'$}
\Text(35,14)[b]{$l't'$}
\end{picture}
$ = \delta^{ss'} \delta^{tt'} \delta(k-k') \delta(l-l') \frac{-\lambda^2}{k+l-\lambda^2 \mu^2+i\epsilon}$
\end{center}
\vspace{10pt}
For the sake of notational simplicity, we dispensed with the sign function. However, we have to keep in mind that the sign of $\epsilon$ may depend on $s$, $t$, $k$, $l$.

From \eqref{eq:chiStar} and \eqref{eq:chiInt} it follows that in the scalar $\phi^3$ model with coupling constant $g$, the vertex is given by
\vspace{-25pt}
\begin{center}
\begin{picture}(50,60)(0,25)
\ArrowLine(10,28)(25, 28)
\ArrowLine(25,32)(10, 32)
\ArrowLine(25,28)(34,15)
\ArrowLine(34,45)(25,32)
\ArrowLine(28,30)(37,43)
\ArrowLine(37,17)(28,30)
\Text(10,26)[t]{$ks$}
\Text(10,34)[b]{$l't'$}
\Text(34,15)[rt]{$k's'$}
\Text(38,15)[lb]{$ju$}
\Text(34,45)[rb]{$lt$}
\Text(38,45)[lt]{$j'u'$}
\end{picture}
$ = g \delta^{ss'} \delta^{tt'} \delta^{uu'} \delta(k-k') \delta(l-l') \delta(j-j').$
\end{center}
\vspace{25pt}

As a first application, we compute the fish graph in the $\phi^3$ model. One finds the following graphs:
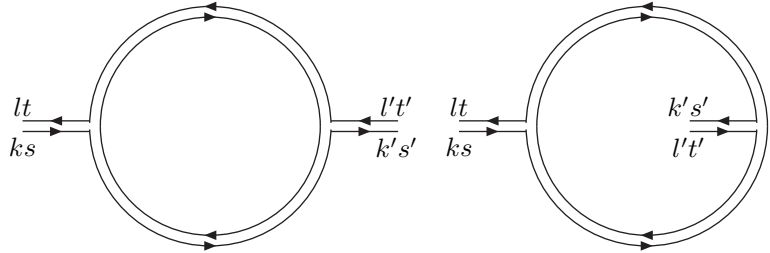
\begin{center}
\begin{picture}(160,100)
\ArrowLine(10,48)(35, 48)
\ArrowLine(35,52)(10, 52)
\ArrowArc(80,50)(45,0,180)
\ArrowArc(80,50)(45,180,0)
\ArrowArcn(80,50)(41,0,180)
\ArrowArcn(80,50)(41,180,0)
\ArrowLine(125,48)(150, 48)
\ArrowLine(150,52)(125, 52)
\CBox(34,49)(36,51){White}{White}
\CBox(124,49)(126,51){White}{White}
\Text(10,46)[t]{$ks$}
\Text(10,54)[b]{$lt$}
\Text(150,46)[t]{$k's'$}
\Text(150,54)[b]{$l't'$}
\end{picture}
\begin{picture}(160,100)
\ArrowLine(10,48)(35, 48)
\ArrowLine(35,52)(10, 52)
\ArrowArc(80,50)(45,0,180)
\ArrowArc(80,50)(45,180,0)
\ArrowArcn(80,50)(41,0,180)
\ArrowArcn(80,50)(41,180,0)
\ArrowLine(96, 48)(121,48)
\ArrowLine(121, 52)(96,52)
\CBox(34,49)(36,51){White}{White}
\CBox(120,49)(122,51){White}{White}
\Text(10,46)[t]{$ks$}
\Text(10,54)[b]{$lt$}
\Text(96,46)[t]{$l't'$}
\Text(96,54)[b]{$k's'$}
\end{picture}
\end{center}
The first one is a planar and the second one a nonplanar\footnote{In the notation of \cite{GrosseWulkenhaar4d} it has genus $g=0$, but two boundary components ($B=2$), and thus a hole.} graph. For the planar part, one obtains
\begin{multline}
\label{eq:PlanarLoop}
i g^2 \lambda^4 \delta(k-k') \delta (l-l') \delta^{ss'} \delta^{tt'} \\
\times \sum_u \int \ud j \ud j' \ \frac{1}{ k + j - \lambda^2 \mu^2 + i \epsilon_1} \frac{1}{ j' + l - \lambda^2 \mu^2 + i \epsilon_2} \left[ \delta(j-j') \right]^2.
\end{multline}
Here we introduced $\epsilon_{1/2}$ in order to remind ourselves that the sign will in general depend on the generalised momenta and the degeneracy indices. Due to the presence of the square of a $\delta$ distribution, this expression is divergent, already before evaluating the loop integral. It is thus no UV divergence in the usual sense (so also a cutoff in the spirit of \cite{GrosseWulkenhaar2d} would not help). That is does not appear in the Euclidean version of the theory can be understood by noting that there, the loop integral is a sum of the form
\[
\sum_{j j'} \delta^{jj'} \delta^{jj'} \frac{1}{n + j + \lambda^2 \mu^2} \frac{1}{j'+m + \lambda^2 \mu^2}.
\]
Here, the square of the Kronecker $\delta$ poses no problems. At the end of this and in the next two sections, we discuss the appearance of this divergence in some detail. But for the moment, we ignore it and 
formally absorb it in a divergent constant $\delta(0)$ such that
$\delta(j-j')^2 = \delta(0) \delta(j-j')$. We can then evaluate the loop integral in \eqref{eq:PlanarLoop}, and in particular discuss the unitarity of the model. We may write it as a convolution:
\begin{equation}
\label{eq:Convolution}
- \int \ud j \ \frac{1}{ j + i \epsilon_1} \frac{1}{ k - l - j - i \epsilon_2}.
\end{equation}
Using
\begin{equation}
\label{eq:x+ie}
\F[\tfrac{1}{x\pm i \epsilon}](p) = -i\sqrt{2\pi}H(\pm p),
\end{equation}
where $\F$ denotes the Fourier transform and $H$ the Heaviside distribution, and\footnote{\label{ft:Renormalisation}Strictly speaking, these products are not well defined in the sense of H\"ormander's product of distributions \cite{Hormander}. Using Steinmann's concept of scaling degree \cite{ScalingDegree}, one can show that the \rhs is, in a certain sense, the unambiguous extension of the product on the \lhs to the singularity at $p=0$.}
\begin{align*}
H(\pm p) H(\pm p) & = H(\pm p), \\ H(\pm p) H(\mp p) & = 0,
\end{align*}
we find that (\ref{eq:Convolution}) vanishes for $\epsilon_1 = \epsilon_2$. In the case $\epsilon_1 = - \epsilon_2$, one obtains
\[
2 \pi i \frac{1}{k-l+i \epsilon_1}.
\]
Because of $\frac{1}{x+i\epsilon} = \text{P}\frac{1}{x} - i \pi \delta(x)$, the imaginary part of (\ref{eq:PlanarLoop}) is thus given by a multiple of
\begin{equation}
\label{eq:ImPlanar}
2 \pi^2 g^2 \lambda^4 \delta(0) \delta(k-k') \delta (l-l') \delta^{ss'} \delta^{tt'} \delta(k- l).
\end{equation}
The multiplicity depends on how $\epsilon_1$ and $\epsilon_2$ behave for the different combinations of $s$, $t$ and $u$ in (\ref{eq:PlanarLoop}).

In the nonplanar graph all generalised momenta are fixed, so there is no loop integral to evaluate. We obtain
\[
i g^2 \lambda^4 \delta(k-l) \delta(k'-l') \delta^{st} \delta^{s't'} \frac{1}{ k + l' - \lambda^2 \mu^2 + i \epsilon_1} \frac{1}{ k' + l - \lambda^2 \mu^2 + i \epsilon_2}.
\]
Using again (\ref{eq:x+ie}), one can show that
\[
\frac{1}{x\pm i\epsilon} \frac{1}{x\pm i\epsilon} = - \del_x \frac{1}{x\pm i\epsilon}
\]
and that the products $\frac{1}{x\pm i\epsilon} \frac{1}{x\mp i\epsilon}$ are not well-defined and have to be renormalised\footnote{For a systematic treatment of renormalisation ambiguities in the products of distributions, we again refer to \cite{ScalingDegree}.}. But even in that case, the product is well-defined on test functions vanishing in a neighborhood of the origin, so renormalisation ambiguities only affect the behavior at $x=0$. We may thus conclude that for $k + k' \neq \lambda^2 \mu^2$ the imaginary part of the nonplanar graph is given by
\begin{equation}
\label{eq:ImNonplanar}
g^2 \lambda^4 \delta(k-l) \delta(k'-l') \delta^{st} \delta^{s't'} \frac{1}{ (k + k' - \lambda^2 \mu^2)^2}.
\end{equation}
For $k + k' = \lambda^2 \mu^2$ there are renormalisation ambiguities when $\epsilon_1 = - \epsilon_2$. But these are not relevant at the moment. The important point is that the contribution (\ref{eq:ImNonplanar}) leads to a violation of unitarity.

We now want to compute two graphs when the internal lines are put on the mass shell (and multiplied by $2 \pi$). For the planar graph, we again find the singularity due to the matching of generalised momenta at the two vertices. Writing this as $\delta(0)$ again, we obtain
\begin{multline*}
2 (2\pi)^2 g^2 \lambda^4 \delta(0) \delta(k-k') \delta (l-l') \delta^{ss'} \delta^{tt'} \int \ud j \ \delta(k + j - \lambda^2 \mu^2) \delta( j + l - \lambda^2 \mu^2). \\
= 8 \pi^2 g^2 \lambda^4 \delta(0) \delta(k-k') \delta (l-l') \delta^{ss'} \delta^{tt'} \delta(k-l).
\end{multline*}
The factor 2 comes from the two-fold degeneracy. This is a multiple of (\ref{eq:ImPlanar}). Thus, it may be possible, by a suitable choice of the sign function, to fulfil the cutting rules. For the nonplanar part, however, we find
\begin{equation*}
4 \pi^2 g^2 \lambda^2 \delta(k-l) \delta(k'-l') \delta^{st} \delta^{s't'} [\delta(k + k' - \lambda^2 \mu^2)]^2.
\end{equation*}
We again find the renormalisation ambiguity at $k + k' = \lambda^2 \mu^2$, but no contribution of the form (\ref{eq:ImNonplanar}). Thus, unitarity is violated in a naive Feynman rules setting.

Let us now come back to the subject of the strange divergence in the planar fish graph.
It is straightforward to see that it is not specific to the $\phi^3$ model, but also shows up in other planar graphs, such as the self-energy in the $\phi^4$ model, or the one-loop correction to the three-point function:
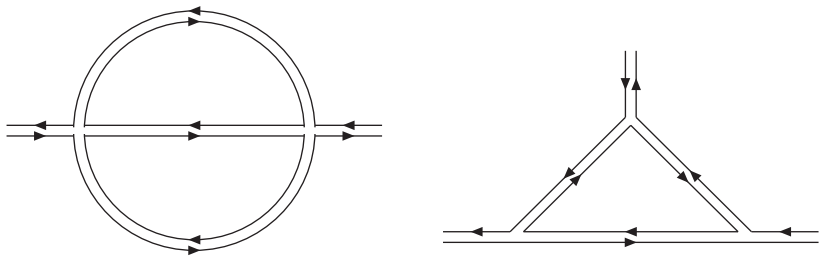
\begin{center}
\begin{picture}(160,100)
\ArrowLine(10,48)(35, 48)
\ArrowLine(35,52)(10, 52)
\ArrowArc(80,50)(45,0,180)
\ArrowArc(80,50)(45,180,0)
\ArrowArcn(80,50)(41,0,180)
\ArrowArcn(80,50)(41,180,0)
\ArrowLine(125,48)(150, 48)
\ArrowLine(150,52)(125, 52)
\ArrowLine(39,48)(121,48)
\ArrowLine(121,52)(39,52)
\CBox(34,49)(36,51){White}{White}
\CBox(124,49)(126,51){White}{White}
\CBox(38,49)(40,51){White}{White}
\CBox(120,49)(122,51){White}{White}
\end{picture}
\begin{picture}(160,100)
\ArrowLine(10,8)(150,8)
\ArrowLine(35,12)(10,12)
\ArrowLine(150,12)(125,12)
\ArrowLine(78,55)(35,12)
\ArrowLine(40,12)(80,52)
\ArrowLine(125,12)(82,55)
\ArrowLine(80,52)(120,12)
\ArrowLine(120,12)(40,12)
\ArrowLine(78,80)(78,55)
\ArrowLine(82,55)(82,80)
\end{picture}
\end{center}

In fact, every ribbon graph in which a closed loop of a single line exists, i.e., every graph that contains a planar subgraph with a loop, is subject to this divergence.

Furthermore, these divergences seem to be present also in a Yang-Feldman quantisation of the model: Then, one of the propagators in the loop integral is replaced by a retarded propagator, and the other one by the Wightman two-point function of the free field \cite{BDFP02, NCDispRel}. The replacement of a Feynman propagator by a retarded one is unessential for the present discussion, since they should differ only in the $i\epsilon$ description at the poles. Due to the broken translation invariance, there is some ambiguity in the definition of the free two-point function, but in any case it has to be a solution to the free field equation and it has to be compatible with the commutator. If the retarded propagator can be written in the form \eqref{eq:Propagator}, as we assumed above, then the commutator (which is derived from the retarded propagator) will conserve the generalised momenta. But then the two-point function must also have a component that conserves the generalised momenta. Thus, the strange planar divergences can not be avoided by using the Yang-Feldman formalism. In the following sections, we will show that they are no artefact of the use of an inappropriate basis, but also appear when the model is discussed in position space.

Finally, we note that the planar divergences will also show up in the four-dimensional case. By a Lorentz transformation, one can always switch to a coordinate system where $\theta$ is of the form
\[
\theta = \begin{pmatrix} \lambda_1^2 \epsilon & 0 \\ 0 & \lambda_2^2 \epsilon \end{pmatrix}.
\]
For the two spatial coordinates that now commute with time, the quadratic potential is the usual harmonic oscillator potential, so the generalised eigenfunctions are given by
\[
\psi_{klmn}^{st}(x) = \chi_{kl}^{st}(x^0, x^1) \phi_{mn}(x^2, x^3),
\]
where $\phi_{mn}$ are Weyl symbols of ket-bras of harmonic oscillator eigenstates. These fulfil
\begin{align*}
\psi_{klmn}^{st} \star \psi_{k'l'm'n'}^{s't'} & = \delta(l-k') \delta_{nm'} \delta^{ts'} \psi_{kl'mn'}^{st'}, \\
\int \ud^4x \ \psi^{st}_{klmn} & = \delta(k-l) \delta_{mn} \delta^{st},
\end{align*}
instead of \eqref{eq:chiStar} and \eqref{eq:chiInt}. The propagator will be of the form
\begin{multline*}
\Delta^{st \ t's'}_{klmn \ l'k'm'n'} = \frac{-1}{ \lambda_1^{-2}(k + l) - \lambda_2^{-2}(m + n) - \mu^2 + i \epsilon \sigma^{st}_{mn}(k, l)} \\ \times \delta(k-k') \delta(l-l') \delta_{mm'} \delta_{nn'} \delta^{ss'} \delta^{tt'}.
\end{multline*}
Since the generalised momenta $k, l$ are conserved by the propagator and at the vertices, one will again find the square of a $\delta$-distribution in the fish graph.

In the setting of the generalised eigenfunctions, the planar divergences arise because the generalised momentum is conserved at the vertices and during propagation. One may thus suspect that the problem is absent when one is not at the self-dual point. In this situation, the description in terms of the generalised eigenfunctions becomes quite complex. However, by switching to position space, one can show that the above reasoning is at least partially correct. As shown in the following sections, the singularity is absent (far enough) below the self-dual point, i.e., if the frequency of the potential is lower than the self-duality frequency, but present at and above the self-dual point.

\section{The retarded propagator}
\label{sec:Delta_ret}

In order to avoid the ambiguities in the definition of a Feynman propagator\footnote{It depends on the quantum state, which is not unique due to the lack of translation invariance.}, we start by considering the retarded propagator. It can be constructed in position space, which avoids the use of the generalised eigenfunctions. In the massless case and without the quadratic potential, the wave operator for a scalar field is given by $\Box = 4 \del_u \del_v$ where we now use
\begin{equation}
\label{eq:uv_neu}
u = x_0 - x_1, \quad v = x_0 + x_1,
\end{equation}
instead of the notation (\ref{eq:uv}) used in Section \ref{sec:Eigenfunctions}. The retarded propagator for this wave operator is 
\[
\Delta_{\text{ret}}(u_1, v_1; u_2, v_2) = \tfrac{1}{2} H(u_1-u_2) H(v_1-v_2),
\]
where $H$ is again the Heaviside distribution. Thus, the square $\Delta_{\text{ret}}^2$ is well-defined without the need for any renormalisation (but see footnote \ref{ft:Renormalisation}). In the presence of the quadratic potential, the retarded propagator will no longer be translation invariant, and the above propagator is multiplied with a function of $u_1, v_1, u_2, v_2$. In the massless case, the wave operator for the Grosse-Wulkenhaar potential is given by, cf. \eqref{eq:WaveOp},
\begin{equation}
\label{eq:WaveOperator_uv}
4 \del_u \del_v + 4 \lambda^{-4} uv
\end{equation}
and we have the following
\begin{proposition}
The retarded propagator for the wave operator \eqref{eq:WaveOperator_uv} is given by
\[
\Delta_{\text{ret}}(u_1, v_1; u_2, v_2) = \tfrac{1}{2} H(u_1-u_2) H(v_1-v_2) \sum_{n=0}^\infty (-1)^n \frac{\left( u_1^2 - u_2^2 \right)^n}{2^n \lambda^{2n} n!} \frac{\left( v_1^2 - v_2^2 \right)^n}{2^n \lambda^{2n} n!}.
\]
\end{proposition}
\begin{proof}
The series on the \lhs has infinite convergence radius and thus yields an analytic function $V(u_1, v_1; u_2, v_2)$. This function (which is in fact a Bessel function) fulfils
\begin{align*}
V(u_1, v_1; u_1, v_1) & = 1, \\
\del_{u_1} V(u_1, v_1; u_2, v_1) & = 0, \\
\del_{v_1} V(u_1, v_1; u_1, v_2) & = 0.
\end{align*}
The first equality assures that when both derivatives in $\del_{u_1} \del_{u_2}$ act on the Heaviside distributions, then one still obtains a $\delta$ distribution for coinciding points. Due to the other two equalities, the mixed terms, where one derivative acts on a Heaviside distribution and the other one on $V$, vanish. Thus, it remains to show that
\[
( \del_{u_1} \del_{v_1} + \lambda^{-4} u_1 v_1 ) V(u_1, v_1; u_2, v_2) = 0,
\]
which is straightforward.
\end{proof}

Before discussing the propagator in more detail, we express it in the coordinates
\begin{subequations}
\label{eq:u_s}
\begin{align}
u_s & = u_1 + u_2, & u_t & = u_1 - u_2,\\
v_s & = v_1 + v_2, & v_t & = v_1 - v_2,
\end{align}
\end{subequations}
as
\begin{align}
\label{eq:Delta_ret_uv}
\Delta_{\text{ret}}(u_s, v_s, u_t, v_t) & = \tfrac{1}{2} H(u_t) H(v_t) \sum_{n=0}^\infty (-1)^n \frac{\left( u_t u_s \right)^n}{2^n \lambda^{2n} n!} \frac{\left( v_t v_s \right)^n}{2^n \lambda^{2n} n!} \\
\label{eq:Delta_ret_J}
& = \tfrac{1}{2} H(u_t) H(v_t) J_0(\lambda^{-2} \sqrt{u_t u_s v_t v_s}).
\end{align}
For imaginary arguments, i.e., for $u_s v_s <0$, the Bessel function diverges as $J_0(ix) \sim e^x/\sqrt{2\pi x}$ \cite[(9.7.1)]{Abramowitz}, which can be seen as the cause of the serious problems we will encounter. Using the inequality
\begin{equation}
\label{eq:Inequality}
2 \betrag{xy} \leq \betrag{x}^2 + \betrag{y}^2
\end{equation}
several times, one finds that the Bessel function is asymptotically bounded by a Gaussian,
\[
\betrag{ J_0(\lambda^{-2} \sqrt{u_t u_s v_t v_s})} \leq C e^{\frac{1}{4 \lambda^2} \left( u_t^2 + u_s^2 + v_t^2 + v_s^2 \right)}.
\]
This is also true for the derivatives. It follows that the retarded propagator is well-defined on test functions that fulfil the bound
\[
\betrag{\del^\beta f} \leq C_\beta e^{-a \left( u_t^2 + u_s^2 + v_t^2 + v_s^2 \right)},
\]
with $a = \frac{1}{4 (\lambda-\epsilon)^2}$, where $\epsilon$ can be chosen arbitrarily small. This is the Gelfand-Shilov space \cite{GelfandShilov} $\mathcal{S}_{\alpha, A}$ where $\alpha$ and $A$ are the quadruples consisting of $\frac{1}{2}$ and $\frac{\sqrt{2} (\lambda- \varepsilon)}{\sqrt{e}}$, respectively, where $\varepsilon$ can be chosen arbitrarily small. For this, we will simply write $S_{\alpha, A}(\R^4)$ with $\alpha = \frac{1}{2}$ and $A = \frac{\sqrt{2} (\lambda- \varepsilon)}{\sqrt{e}}$ in the following. By the above reasoning, $\Delta_{\text{ret}}$ can be interpreted as an element of ${S'}_{\alpha, A}(\R^4)$

\begin{remark}
\label{rem:EffMass}
We recall that for a massive theory (without quadratic potential) the retarded propagator is given by
\[
\tfrac{1}{2} H(u_t) H(v_t) J_0(\mu \sqrt{u_t v_t}).
\]
One thus has the very natural interpretation of the propagator \eqref{eq:Delta_ret_J} as the one for a position dependent mass $\mu^2 = \lambda^{-4} u_s v_s$. This is the value of the potential at the center of mass of the two points $(u_1, v_1)$ and $(u_2, v_2)$. The problems for $u_s v_s < 0$ stem from the fact that the model becomes tachyonic (and ever more so as $u_s v_s \to -\infty$).
\end{remark}
The Fourier transform of \eqref{eq:Delta_ret_uv} (which can be interpreted as an element of ${\mathcal{S}'}^{\alpha, A}(\R^4)$ with $\alpha$ and $A$ as above) is
\begin{equation}
\label{eq:Delta_ret}
\hat{\Delta}_{\text{ret}}(k_s, l_s, k_t, l_t) = - \tfrac{1}{2} \sum_{n=0}^\infty \left(\tfrac{-1}{4 \lambda^4}\right)^{n} \delta^{(n)}(k_s) \delta^{(n)}(l_s) \left( \tfrac{1}{k_t-i\epsilon} \tfrac{1}{l_t-i\epsilon}\right)^{n+1}.
\end{equation}
Here $k_{s/t}$ is the Fourier dual of $u_{s/t}$ and $l_{s/t}$ that of $v_{s/t}$. The appearance of derivatives of the $\delta$-distribution in $k_s$ and $l_s$ indicates that momentum is not conserved.
In Appendix \ref{sec:Fourier} it is shown that the Fourier transform of \eqref{eq:Delta_ret_J} can be expressed in terms of the Bessel function $K_0$. 

\section{Planar divergences}
\label{sec:PlanarDiv}

Having the retarded propagator at hand, we can now discuss the origin of the planar divergences found in Section \ref{sec:Causality}. In the setting of the naive Feynman rules, the planar $\phi^3$ fish graph is given by
\[
\Delta_F(x,y) \star_x \bar{\star}_y \Delta_F(x,y),
\]
where $\bar{\star}$ denotes the $\star$-product with $\theta$ replaced by $- \theta$. In the Yang-Feldman approach we would have to compute similar products, but with one of the propagators replaced by $\Delta_{\text{ret}}$ and the other one by one of the two-point functions $\Delta_{\pm}$. This, however, requires the choice of a state, which we would like to avoid for as long as possible. We thus try to compute the product
\begin{equation}
\label{eq:RetPlanar}
\Delta_{\text{ret}}(x,y) \star_x \bar{\star}_y \Delta_{\text{ret}}(x,y).
\end{equation}
Even though it has no direct physical significance, the study of this product helps to understand the origin of the planar divergences. We want to compute this product in the coordinates $u_{s/t}$, $v_{s/t}$. By using $[u,v] = 2 i \lambda^2_{\text{nc}}$, we obtain the commutation relations
\begin{equation*}
[u_s, v_s]^- = [u_t, v_t]^- = [u_s, u_t]^- = [v_s, v_t]^- = 0, \quad [u_s, v_t]^- = [u_t, v_s]^- = 4 i \lambda^2_{\text{nc}},
\end{equation*}
where $[\cdot, \cdot]^-$ is the commutator where the $\bar{\star}$-commutator was used in the second argument. Thus, the correct twisting factor for our momenta is, cf. (\ref{eq:TwistedConvolution}),
\begin{equation}
\label{eq:Twisting}
e^{-2 i \lambda_{\text{nc}}^2 (k_s \tilde{l}_t + k_t \tilde{l}_s - l_s \tilde{k}_t - l_t \tilde{k}_s)}.
\end{equation}

Now the question is in which sense the product \eqref{eq:RetPlanar} should be defined. As already noted, even the pointwise (commutative) product is not well-defined in the sense of H\"ormander's product of distributions. To our mind, the most conservative approach to a definition of \eqref{eq:RetPlanar} is the following\footnote{In \cite{Quasiplanar}, this strategy was pursued for the definition of quasiplanar Wick products.}: In order to disentangle the problems connected to the $\star$-product and the distributional character of the retarded propagator, one begins by defining the $\star$-product at different points. In the next step, one checks whether the limit of coinciding points makes sense. The definition of the $\star$-product at different points can be done by duality, as proposed in \cite{Soloviev}:
\begin{equation}
\label{eq:Duality}
\skal{\Delta_{\text{ret}} \otimes_{\star_x \bar{\star}_y} \Delta_{\text{ret}}}{f \otimes g} = \skal{\Delta_{\text{ret}} \otimes \Delta_{\text{ret}}}{f \otimes_{\star_x \bar{\star}_y} g}.
\end{equation}
Here we wrote the planar $\star$-product at different points in the form of a tensor product. Using \eqref{eq:Twisting}, we have
\begin{multline}
\label{eq:PlanarDifferentPoints}
(f \otimes_{\star_x \bar{\star}_y} g) \hat{\ } (k_s, l_s, k_t, l_t; \tilde{k}_s, \tilde{l}_s, \tilde{k}_t, \tilde{l}_t) \\ = e^{-2 i \lambda_{\text{nc}}^2 (k_s \tilde{l}_t + k_t \tilde{l}_s - l_s \tilde{k}_t - l_t \tilde{k}_s)} \hat{f}(k_s, l_s, k_t, l_t) \hat{g}(\tilde{k}_s, \tilde{l}_s, \tilde{k}_t, \tilde{l}_t).
\end{multline}
Formally, this may be written as
\begin{multline}
\label{eq:PlanarDifferentPointsFormal}
(f \otimes_{\star_x \bar{\star}_y} g)(u_s, v_s, u_t, v_t; \tilde{u}_s, \tilde{v}_s, \tilde{u}_t, \tilde{v}_t) \\ = e^{2i\lambda_{\text{nc}}^2(\del_{u_s} \del_{\tilde{v}_t} + \del_{u_t} \del_{\tilde{v}_s} - \del_{v_s} \del_{\tilde{u}_t} - \del_{v_t} \del_{\tilde{u}_s})} f(u_s, v_s, u_t, v_t) g(\tilde{u}_s, \tilde{v}_s, \tilde{u}_t, \tilde{v}_t).
\end{multline}
In order for the \rhs of \eqref{eq:Duality} to be well-defined, we have to require the \rhs of \eqref{eq:PlanarDifferentPoints} to be an element of $\mathcal{S}^{\alpha, A}(\R^8)$ (or the \rhs of \eqref{eq:PlanarDifferentPointsFormal} to be an element of $\mathcal{S}_{\alpha, A}(\R^8)$).
For this, we might have to choose $f$ and $g$ from a suitable subset of $\mathcal{S}_{\alpha, A}(\R^4)$. That this is possible if one is far enough below the self-dual point is the result of the following\footnote{Similar considerations on the $\star$-product of elements of $\mathcal{S}^{\beta, B}$ can be found in \cite{Chaichian}.}

\begin{proposition}
\label{prop:below}
For $\alpha=\frac{1}{2}$, $A=\frac{\sqrt{2} (\lambda-\varepsilon)}{\sqrt{e}}$ and $\sqrt{e} \lambda_{\text{nc}} <( \lambda-\varepsilon)$, there is a nontrivial subset $\mathcal{S}$ of $\mathcal{S}_{\alpha, A}(\R^4)$, such that, for $f, g \in \mathcal{S}$, the \rhs of \eqref{eq:PlanarDifferentPointsFormal} is well-defined as an element of $\mathcal{S}_{\alpha, A}(\R^8)$. More precisely, this is the case for $\mathcal{S} = \mathcal{S}^{\beta, B}_{\alpha, A}(\R^4)$ with $\beta = \frac{1}{2}$, $B = \frac{1}{\sqrt{2e} (\lambda-\varepsilon)}$.
\end{proposition}
\begin{proof}
According to \cite{GelfandShilov}, the operator $f(\del)$ for an entire function $f$ of order less than or equal to $\frac{1}{\beta}$ and type less than $\frac{\beta}{B^{1/\beta}e^2}$ is well-defined on the space $\mathcal{S}^{\beta, B}_{\alpha, A}$ and maps it to the space $\mathcal{S}^{\beta, B e^\beta}_{\alpha, A}$. Using \eqref{eq:Inequality}, it is easy to see that the twisting in \eqref{eq:PlanarDifferentPointsFormal} has order $2$ and type $\lambda_{\text{nc}}^2$. It follows that for $\beta=\frac{1}{2}$ and $B = \frac{1}{\sqrt{2e} (\lambda-\varepsilon)}$, the \rhs of \eqref{eq:PlanarDifferentPointsFormal} is well-defined as an element of $\mathcal{S}_{\alpha, A'}(\R^8)$, provided that $\sqrt{e} \lambda_{\text{nc}} <( \lambda-\varepsilon)$. It remains to show that the space $\mathcal{S}^{\beta, B}_{\alpha, A}(\R^4)$ is nontrivial. As can easily be seen by considering a Gaussian, this space, for $\alpha=\beta=\frac{1}{2}$, is nontrivial provided that $AB \geq 1/e$, which is fulfilled.
\end{proof}

By applying more sophisticated methods, it might be possible to get rid of the factor $\sqrt{e}$ in the restriction on $\lambda$. This would mean that to be below the self-dual point is a sufficient condition for the possibility to define the planar $\star$-product at different points for elements of ${\mathcal{S}'}_{\alpha, A}(\R^4)$. But in any case we can show that it is a necessary condition:

\begin{proposition}
For $\lambda_{\text{nc}} \geq \lambda$, $\alpha=\frac{1}{2}$ and $A=\frac{\sqrt{2} (\lambda-\varepsilon)}{\sqrt{e}}$, there are no $\varepsilon, \varepsilon'>0$ such that there are nontrivial $\hat{f}, \hat{g} \in \mathcal{S}^{\alpha, A'}(\R^4)$, with $A'=\frac{\sqrt{2} (\lambda-\varepsilon')}{\sqrt{e}}$ for which the \rhs of \eqref{eq:PlanarDifferentPoints} is an element of $\mathcal{S}^{\alpha, A}(\R^8)$.
\end{proposition}
\begin{proof}
We assume that such $\varepsilon, \varepsilon'$ and such functions $\hat f$, $\hat g$ exist. Now, according to \cite{GelfandShilov}, elements of $\mathcal{S}^{\beta, B}$, for $\beta<1$, are entire functions that fulfil the bound
\[
\betrag{f(x+iy)} \leq C e^{b\betrag{y}^{\frac{1}{1-\beta}}}
\]
with $b=\frac{1-\beta}{e}(Be)^{\frac{1}{1-\beta}} +\delta$, where $\delta$ can be chosen arbitrarily small. In our case, this means that $\hat f$ and $\hat g$ are entire functions that fulfil the bounds
\begin{equation}
\label{eq:fgbounds}
\betrag{\hat f(x+iy)} \leq c e^{b' \betrag{y}^2}, \quad \betrag{\hat g(x+iy)} \leq c' e^{b' \betrag{y}^2},
\end{equation}
with $b' = (\lambda - \varepsilon')^2 + \delta'$. Here $x$ and $y$ are elements of $\R^4$. Since we assumed that the \rhs of \eqref{eq:PlanarDifferentPoints} is an element of $\mathcal{S}^{\alpha, A}(\R^8)$, also the bound
\[
\betrag{(f \otimes_{\star_x \bar{\star}_y} g) \hat{\ } (x+iy; \tilde{x}+i\tilde{y})} \leq C e^{b \left( \betrag{y}^2 + \betrag{\tilde{y}}^2 \right)},
\]
with $b=(\lambda-\varepsilon)^2 +\delta$, has to be fulfilled. We define the matrix
\[
\gamma = \begin{pmatrix} 0 & 0 & 0 & -1 \\ 0 & 0 & 1 & 0 \\ 0 & -1 & 0 & 0 \\ 1 & 0 & 0 & 0 \end{pmatrix}.
\]
The above inequality then leads to
\[
e^{2 \lambda_{\text{nc}}^2 \left( \betrag{x}^2 - \betrag{y}^2 \right)} \betrag{ \hat{f}(z) \hat{g}(i\gamma z)} \leq C e^{b \left( \betrag{y}^2 + \betrag{x}^2 \right)},
\]
where $z = x + i y$. Thus, we have
\begin{equation}
\label{eq:Ineq_1}
\betrag{ \hat{f}(z) \hat{g}(i\gamma z)} \leq C e^{- \left( 2 \lambda_{\text{nc}}^2 - b \right) \betrag{x}^2 + \left( 2 \lambda_{\text{nc}}^2 + b \right) \betrag{y}^2}.
\end{equation}
Because of $\lambda_{\text{nc}} \geq \lambda$, $\delta$ can be chosen such that $2 \lambda_{\text{nc}}^2 - b$ is positive. Thus, the function $\hat{f}(z) \hat{g}(i\gamma z)$ falls off with order $2$ and type $2 \lambda_{\text{nc}}^2 - b$ in the real direction. On the other hand, from the bounds \eqref{eq:fgbounds} on $\hat f$ and $\hat g$, it follows that
\begin{equation}
\label{eq:Ineq_2}
\betrag{ \hat{f}(z) \hat{g}(i\gamma z)} \leq C' e^{ b' \left( \betrag{x}^2 + \betrag{y}^2 \right)}.
\end{equation}
Thus, $\hat{f}(z) \hat{g}(i\gamma z)$ has growth of order $2$ and type $b'$ in the imaginary direction.
Hence, for $2 \lambda_{\text{nc}}^2 - b > b'$, the entire function $F(z) = e^{b'z^2} \hat{f}(z) \hat{g}(i\gamma z)$ is bounded on the real and imaginary axes and, by \eqref{eq:Ineq_2}, uniformly bounded by $C' e^{2 b' \betrag{z}^2}$ in between the axes. By the Phragm\'en-Lindel\"of principle \cite[Thm 2.5.2]{Evgrafov}, it is thus bounded on the whole complex plane, and can only be a constant. But since its limit in the real direction is $0$, cf. \eqref{eq:Ineq_1}, we have $F(z) = 0$. Thus, nontrivial $\hat f$, $\hat g$ can only exist for $2 \lambda_{\text{nc}}^2 - b \leq b'$. This means $2\lambda_{\text{nc}}^2 - (\lambda-\varepsilon)^2 - (\lambda-\varepsilon')^2 \leq \delta + \delta'$.
It is clear that for $\lambda_{\text{nc}} \geq \lambda$, $\delta$ and $\delta'$ can not be chosen arbitrarily small, contrary to our assumption.
\end{proof}

By this analysis, it is not possible to define the planar product of two retarded propagators at different points at or above the self-dual point. Thus, the problem in defining the product \eqref{eq:RetPlanar} does not stem from the limit of coinciding points, since already the product at different points is ill-defined. In this sense, this is no UV divergence.

The result of this rather abstract argument can be checked with a concrete calculation. We will do this in a formal way, i.e., we use the series form \eqref{eq:Delta_ret_uv} (or equivalently \eqref{eq:Delta_ret}) of the retarded propagator and compute the planar product for the individual terms. The hope is that we obtain a (power) series that can again be summed up.
The twisted convolution that we want to compute is then\footnote{Since our calculation is formal anyway, we could also compute the $\star$-product with the formal series \eqref{eq:StarProduct}. The result is the same. But since a calculation in momentum space is needed later on, we chose to present it in terms of the twisted convolution \eqref{eq:TwistedConvolution}.}
\begin{multline*}
\tfrac{1}{(4\pi)^2} \sum_{mn} \int \ud k_s \ud l_s \ud k_t \ud l_t \ e^{-2 i \lambda^2_{\text{nc}} (k_s \tilde{l}_t + k_t \tilde{l}_s - l_s \tilde{k}_t - l_t \tilde{k}_s)} \left(\tfrac{-1}{4 \lambda^4}\right)^{m+n} \\ \times \delta^{(m)}(k_s) \delta^{(m)}(l_s) \left( \tfrac{1}{k_t-i\epsilon} \tfrac{1}{l_t-i\epsilon}\right)^{m+1} \\ \times \delta^{(n)}(\tilde{k}_s - k_s) \delta^{(n)}(\tilde{l}_s - l_s) \left( \tfrac{1}{\tilde{k}_t - k_t-i\epsilon} \tfrac{1}{\tilde{l}_t - l_t-i\epsilon}\right)^{n+1}.
\end{multline*}
We now use
\[
\delta^{(n)}(x-y) f(x) = \sum_{m=0}^n (-1)^m \binom{n}{m} f^{(m)}(y) \delta^{(n-m)}(x-y).
\]
to get rid of the twisting factor, by first applying this equality to $\tilde{k}_s$ and $\tilde{l}_s$ and then to $k_s$ and $l_s$. One thus obtains
\begin{multline*}
\tfrac{1}{(4\pi)^2} \sum_{m,n = 0}^\infty \sum_{j_1, j_2 = 0}^n \sum_{j_3, j_4 = 0}^m \binom{n}{j_1} \binom{n}{j_2} \binom{m}{j_3} \binom{m}{j_4} \frac{(2i\lambda^2_{\text{nc}})^{\sum j_i}}{(-4 \lambda^4)^{m+n}} (-1)^{j_1+j_4} \\
\times \delta^{(m+n-j_1-j_3)}(\tilde{k}_s) \delta^{(m+n-j_2-j_4)}(\tilde{l}_s) \int \ud k_t \ud l_t \ \left( \tfrac{1}{k_t-i\epsilon} \right)^{m-j_2+1} \left( \tfrac{1}{l_t-i\epsilon}\right)^{m-j_1+1} \\ \times \left( \tfrac{1}{\tilde{k}_t - k_t-i\epsilon} \right)^{n-j_4+1} \left( \tfrac{1}{\tilde{l}_t - l_t-i\epsilon}\right)^{n-j_3+1}.
\end{multline*}
Let us consider the integral over $k_t$. For $m-j_2 \geq 0$ and $n-j_4 \geq 0$, the integral yields, in position space, a multiple of
\[
\tfrac{i^{m+n-j_2-j_4}}{(m-j_2)! (n-j_4)!} H(u_t) u_t^{m+n-j_2-j_4}.
\]
For $m-j_2<0$ one has $n-j_4 \geq j_2-m > 0$ and the integral yields, in position space, a multiple of the product of $\delta^{(j_2-m-1)}(u_t)$ and $H(u_t) u_t^{n-j_4}$. Albeit this product is not well-defined in the sense of H\"ormander, it vanishes in the sense of Steinmann's scaling degree. The analogous argument works for $n-j_4<0$. We thus obtain, in position space, a multiple of
\begin{multline*}
\tfrac{1}{4} H(u_t) H(v_t) \sum_{m,n = 0}^\infty \sum_{j_i=0}^{\min(m,n)} \tfrac{n!n!m!m!}{\prod_i j_i! (m-j_i)! (n-j_i)!} \frac{(2\lambda^2_{\text{nc}})^{\sum j_i}}{(4 \lambda^4)^{m+n}} (-1)^{j_1+j_4} \\ \times (u_t v_s)^{m+n-j_2-j_4} (u_s v_t)^{m+n-j_1-j_3}
\end{multline*}
Apart from the factor $(-1)^{j_1 + j_4}$ the summand is invariant under the exchanges $j_1 \leftrightarrow j_3$ and $j_2 \leftrightarrow j_4$. It follows that only terms where $j_2+j_4$ and $j_1+j_3$ are even contribute. We can thus write the above as
\begin{equation}
\label{eq:ProductSeries}
\tfrac{1}{4} H(u_t) H(v_t) \sum_{m,n = 0}^\infty \sum_{k,l=0}^{\min(m,n)} c^{k}_{mn} c^{l}_{mn} \frac{(4\lambda^4_{\text{nc}})^{k+l}}{(4 \lambda^4)^{m+n}} (u_t v_s)^{m+n-2k} (u_s v_t)^{m+n-2l}
\end{equation}
with
\[
c_{mn}^k = \sum_{j=\max(0,2k-\min(m,n))}^{\min(2k,\min(m,n))} \tfrac{n!m!}{j! (2k-j)! (m-j)! (m-2k+j)! (n-j)! (n-2k+j)!} (-1)^j.
\]
For $m=n=k$, one finds $c_{mm}^m = (-1)^m$. While it seems to be hard make a statement on the convergence of the series in \eqref{eq:ProductSeries} for fixed but general $u_t v_s$ and $u_s v_t$, it is easy to show that it does not converge for $u_s = v_s = 0$. This also shows that it does not converge as a power series, contrary to \eqref{eq:Delta_ret_uv}, since the zeroth order coefficient does not converge. For $u_s = v_s = 0$ we only get a contribution for $m = n = k=l$, so that the above series reduces to\footnote{For the coefficient of the $u_t u_s v_t v_s$ component, one finds, using $c^m_{m m+1} = (m+1) (-1)^m$, the series $\sum_m (m+1)^2 (\lambda_{\text{nc}}/\lambda)^{8m}$, which diverges even worse.}
\begin{equation}
\label{eq:divSeries}
\sum_{m=0}^\infty \left( \frac{\lambda^4_{\text{nc}}}{\lambda^4} \right)^{2m}.
\end{equation}
Obvioulsy, this diverges unless $\lambda_{\text{nc}} < \lambda$, so the planar product (\ref{eq:RetPlanar}) is only well-defined below the self-dual point.
Using \eqref{eq:lambda}, the above reduces to
\[
\sum_{m=0}^\infty \Omega^{4m} = \frac{1}{1-\Omega^4}.
\]
For $\Omega=1-\epsilon$, we thus find a divergence $\epsilon^{-1}$ as $\epsilon \to 0$.
Note that the problem is not that the loop integral over the momenta diverges. Thus, this is no ordinary UV divergence, similarly to what we found previously.

\begin{remark}
Let us consider what happens in the case of a massive field. Then the series in \eqref{eq:Delta_ret_uv} will be a power series in $\lambda^{-4n}$ and $\mu^2$. However, at zeroth order in $\mu^2$, one finds again \eqref{eq:Delta_ret_uv}. Thus, when one calculates the planar product in a formal way, then a nonvanishing mass does not help. If one does not resort to a formal calculation, then it is to be expected that the propagator is still only defined as an element of ${\mathcal{S}'}_{\alpha, A}(\R^4)$ with $\alpha$ and $A$ as above, since in the long range the quadratic potential will always dominate the mass, cf. Remark \ref{rem:EffMass}.
\end{remark}

The argument given up to now is not complete in the sense that we computed the planar square (\ref{eq:RetPlanar}) of the retarded propagator, which is not what appears in actual loop calculations.
In the Yang-Feldman formalism on noncommutative spacetimes, the planar fish graph loop integral is of the form \cite{BDFP02,NCDispRel}
\begin{equation}
\label{eq:YFLoop}
\Delta_{+}(x,y) \star_x \bar{\star}_y \Delta_{\text{ret}}(x,y) + \Delta_{-}(x,y) \bar{\star}_x \star_y \Delta_{\text{ret}}(x,y),
\end{equation}
where $\Delta_+$ is the Wightman two-point function and $\Delta_-(x,y) = \Delta_+(y,x)$. Thus, one has to choose a state, which, however, is not unique due to the lack of translation invariance. But usually the two-point function is defined on the same test function space as the retarded propagator (or a subset thereof). Thus, by the analysis in the beginning of this section (which only used the structure of the test function space), we expect the same problems as above. In order to be more concrete, we choose a particular two-point function and repeat the (formal) calculation from above.

A two-point function $\Delta_+$ has to be compatible with the commutator, which is defined via the retarded propagator, i.e.,
\begin{multline*}
\hat{\Delta}_{+}(k_s, l_s, k_t, l_t) - \hat{\Delta}_{+}(k_s, l_s, -k_t, -l_t) \\ = i\left( \hat{\Delta}_{\text{ret}}(k_s, l_s, k_t, l_t) - \hat{\Delta}_{\text{ret}}(k_s, l_s, -k_t, -l_t) \right).
\end{multline*}
Furthermore, it has to be a solution of the field equation. Finally, some kind of positivity would be nice. Thus, ignoring the usual infrared problems\footnote{By restricting to test functions that vanish in a neighborhood of $k_t = l_t = 0$.},
\begin{multline*}
\hat{\Delta}_{+}(k_s, l_s, k_t, l_t) = \tfrac{\pi}{2} \sum_{n=0}^\infty  \tfrac{1}{4^n \lambda^{4n} n!} \delta^{(n)}(k_s) \delta^{(n)}(l_s) \\ \times \left[ \left(k_t\right)^{-n-1}_+ \delta^{(n)}(l_t) + \left(l_t\right)^{-n-1}_+ \delta^{(n)}(k_t) \right].
\end{multline*}
would be a suitable two-point function, cf. (\ref{eq:Delta_ret}). Using this two-point function, we compute the product \eqref{eq:YFLoop} at $u_s = v_s = 0$, i.e., we consider the component where all derivatives of the $\delta$ distributions of $k_s$, $l_s$, $\tilde{k}_s$ and $\tilde{l}_s$ are shifted on the twisting. For given $m \geq n$ we obtain
\begin{multline*}
\tfrac{\pi (-1)^{m+n}}{(4\pi)^2(m-n)!} \left(\tfrac{\lambda_{\text{nc}}}{\lambda} \right)^{4(m+n)} \int \ud k_t \ud l_t \ \left( \tfrac{1}{\tilde{k}_t - k_t-i\epsilon} \tfrac{1}{\tilde{l}_t - l_t-i\epsilon}\right)^{n-m+1} \\
\times \left[ \left( \left(k_t\right)^{n-m-1}_+ + (-1)^{m+n} \left(k_t\right)^{n-m-1}_- \right) \delta^{(m-n)}(l_t) \right. \\
\left. + \left( \left(l_t\right)^{n-m-1}_+ + (-1)^{m+n} \left(l_t\right)^{n-m-1}_- \right) \delta^{(m-n)}(k_t) \right].
\end{multline*}
Here we used 
\[
x^n \delta^{(m)}(x) = (-1)^n n! \binom{m}{n} \delta^{(m-n)}(x).
\]
Shifting the derivatives \wrt $l_t$ and $k_t$ away from the $\delta$ distributions, we see that this vanishes for $m>n$. For $m < n$, we also get a vanishing expression, since it involves the products $l_t^n \delta^{(m)}(l_t)$ and $k_t^n \delta^{(m)}(k_t)$. Thus, only the contributions with $m=n$ survive. As above, these are independent of $m$,\footnote{The integral itself is UV finite, but has the usual infrared problems.} so we again find the series (\ref{eq:divSeries}), which diverges for $\lambda_{\text{nc}} \geq \lambda$.

\begin{remark}
Even if one is not at the self-dual point, one still has, under a suitable exchange of positions and momenta, the duality $S[\phi, m, \Omega] \to \Omega^2 S[\phi, m \Omega^{-1}, \Omega^{-1}]$, cf. \cite{Ghost}. Thus, one might wonder about the compatibility of this fact with the above finding that the model behaves well for $\Omega<1$ but diverges otherwise. The  point is that the above duality does not respect crucial properties of quantum field theories on Minkowski space, such as causality or positive energy. In particular, the retarded propagator is not invariant under the above duality transformation.
\end{remark}

\begin{remark}
Also on Euclidean space divergences at the self-dual point were found, namely in planar tadpoles of the Gross-Neveu model \cite[App. A.4]{VignesTourneret}. Similar to the findings presented here, the singularity is present even before the loop integral is evaluated. Also there, the origin is the behaviour of the propagator for large spatial distances\footnote{The author would like to thank F.~Vignes-Tourneret for private communication on his work. He would also like to thank the referee for pointing out this reference.}. It would be interesting to further study the similarities of the two effects. 
\end{remark}

\begin{remark}
Finally, a comment on the $\phi^4$ model. The fish graph loop calculated here also occurs in the four-point function of the $\phi^4$ model, so the problem is not specific to the $\phi^3$ model. For the two-loop self-energy graph shown in Section~\ref{sec:Causality} one has to compute products involving a retarded propagator and two of the two-point functions $\Delta_\pm$. If already the product of $\Delta_{\text{ret}}$ with one of these does not exist, then neither do the higher order products.
\end{remark}

\section{Summary \& Outlook}
We discussed noncommutative field theory with Grosse-Wulkenhaar potential on the the two-dimensional Minkowski space in two ways: In the first approach, we restricted ourselves to the self-dual point and used a continuous set of generalised eigenfunctions of the wave operator. This we used to postulate naive Feynman rules. In this setting, we found a new type of divergence in the planar sector. By considering the situation in position space, we showed that this divergence is not due to an inappropriate choice of the basis. Instead, the fast growth of the propagator in some directions makes the definition of the planar $\star$-product impossible, even before considering the limit of coinciding points.

In our opinion, the appearance of this new type of divergences is an interesting phenomenon that deserves more detailed studies. These could proceed along the following lines: In order to relate the two approaches discussed here (position space and generalised momenta), it would be useful to have a representation of the retarded propagator in terms of the eigenfunctions $\chi_{kl}^{st}$. This would amount to find an appropriate sign function $\sigma_{st}(k,l)$.

Furthermore, it would be interesting to know whether one can get rid of the factor $\sqrt{e}$ in the restriction on $\lambda$ in Proposition~\ref{prop:below}. This would mean that the model is well-defined on the whole interval $\Omega \in [0,1)$. Otherwise, the self-dual point may not be so special after all. Another (possibly related) question is the following: We have shown that at and above the self-dual point the individual terms of a series expansion of the planar square of the retarded propagator diverge. We conjecture that below the self-dual point all indivual terms in this expansion converge. If this is the case, it remains to check whether the series as a whole converges below the self-dual point.

One could also study the model in the matrix basis of Grosse and Wulkenhaar. In the Minkowski case, the propagator will then take a more complicated form than in the Euclidean case, but it might still be possible check whether one runs into problems similar to those discussed here.

Another important point is renormalisation. Because of the uncommon type of the divergences, it is not clear whether such a program can be successful and how one should proceed, but perhaps the formal renormalisation used in Section~\ref{sec:Causality} would be a good starting point.

Finally, one should treat the four-dimensional case. As discussed in Section \ref{sec:Causality}, the use of the generalised eigenfunctions will again lead to divergences in planar (sub)graphs. Preliminary results suggest that this is also true in position space, at least in a formal sense. In that case, it would be important to understand why these problems are absent in the Euclidean setting.

\begin{center}
{\bf Acknowledgements}
\end{center}
The author would like to thank Dorothea Bahns, Andr\'e Fischer, Harald Grosse and Olaf Lechtenfeld for stimulating discussions and valuable comments.

\appendix

\section{The relation to the matrix model}
\label{app:FischerSzabo}
We now want to clarify the connection to the matrix model setting proposed by Fischer and Szabo \cite{FischerSzabo}. They work at the self-dual point and also consider the eigenfunctions $\chi_\pm$, $\eta_\pm$, cf. \eqref{eq:chi_eta}. However, they suggest to transform the model to matrix form by considering the Gelfand triple
\[
\mathcal{S}_\alpha^\alpha(\R) \subset L^2(\R) \subset \mathcal{S}_\alpha^\alpha(\R)',
\]
where $\mathcal{S}_\alpha^\alpha(\R)$ is a Gelfand-Shilov space \cite{GelfandShilov}. For elements of this space, they claim (their Theorem 4.2),
\begin{eqnarray*}
\lim_{k \to -i \infty} \braket{\chi^k_\pm}{\phi} = 0 \quad \forall \phi \in \mathcal{S}_\alpha^\alpha(R), \\
\lim_{k \to +i \infty} \braket{\eta^k_\pm}{\phi} = 0 \quad \forall \phi \in \mathcal{S}_\alpha^\alpha(R),
\end{eqnarray*}
where the limit is taken in the lower (upper) complex half-plane. Furthermore, the eigenfunctions $\chi^k_\pm$ and $\eta^k_\pm$ have poles at $k = - 2 i (2n + 1)$ and $k = + 2 i (2n + 1)$, respectively, for $n \in \mathbb{N}_0$. The corresponding residues are given by
\begin{align*}
\mathop{\text{Res}}_{k = - 2 i (2n + 1)} [\chi^k_\pm(q)] & \propto f_n^-(q) \propto e^{+i\frac{q^2}{2}} H_n(\sqrt{-i}q), \\
\mathop{\text{Res}}_{k = + 2 i (2n + 1)} [\eta^k_\pm(q)] & \propto f_n^+(q) \propto e^{-i\frac{q^2}{2}} H_n(\sqrt{+i}q),
\end{align*}
where $H_n$ are the Hermite polynomials. Finally, one has
\[
\overline{\chi_k^\pm}(q) = \eta_k^\pm(q).
\]
From these facts they conclude (their Corollary 4.3), that, by closing the contour of integration\footnote{To be precise, in the mentioned Corollary, they write $\phi = \frac{1}{2} \sum_s \sum_n \ket{f_n^{-s}} \braket{f_n^s}{\phi}$. The second term seems to be added for convenience.}
\[
\phi = \sum_{s} \int \ud k \ \ket{\chi^k_s} \braket{\chi^k_s}{\phi} = \sum_n \ket{f_n^-} \braket{f_n^+}{\phi}.
\]
Since the $f_n^-$ are neither elements of $\mathcal{S}_\alpha^\alpha(\R)$ nor of $L^2(\R)$, the convergence seems to be in $\mathcal{S}_\alpha^\alpha(\R)'$, but it is not clear in which topology.

Now the $\phi$ in the above equation is still only a ket. However, for $\varphi = \ket{\phi}\bra{\psi} \in \mathcal{S}_\alpha^\alpha (\R) \otimes \mathcal{S}_\alpha^\alpha (\R)$, one obtains the expansion
\[
\varphi = \sum_{mn} \varphi_{m n} f_{mn}
\]
with
\[
\varphi_{mn} = \braket{f_n^+}{\phi} \braket{\psi}{f_n^-} \in \mathbbm{C}
\]
and
\[
f_{mn} = \ket{f_m^-}\bra{f_n^+}.
\]
The $f_{mn}$ fulfil the usual properties of a matrix base, i.e.,\footnote{However, it is not clear in which sense these relations should be understood, since the $f_n^\pm$ are not in $L^2(\R)$. Also an interpretation in a distributional sense as for \eqref{eq:chiStar}, \eqref{eq:chiInt} is not possible, since one can not interpret $f_n^\pm$ as a distribution in $n$, due to the discreteness of the imaginary eigenvalues.}
\[
f_{mn} f_{m'n'} = \delta_{n m'} f_{mn'}; \quad \Tr f_{mn} = \delta_{mn}.
\]
One can thus use them to bring the model into matrix form and treat it similarly to \cite{GrosseWulkenhaar2d}. One then arrives at the following representation of the propagator\footnote{Using the basis $\eta$ as a starting point, one would arrive at a similar propagator where the denominator is replaced by its complex conjugate.} (equation (3.58) of \cite{FischerSzabo} with $\sigma = \frac{1}{2}$):
\begin{equation}
\label{eq:PropFischerSzabo}
\Delta_{mn \ n'm'} = \delta_{m m'} \delta_{n n'} \frac{-\lambda^2}{ - 4i (m+n+1) + \lambda^2 \mu^2} 
\end{equation}

However, some remarks are in order: As mentioned in \cite{FischerSzabo}, the use of Gelfand-Shilov spaces as test function spaces for noncommutative field theories has been proposed by several authors \cite{Chaichian, Soloviev}. This would imply that fields are elements of the dual space $\mathcal{S}_\alpha^\alpha(\R^2)'$. In the setting of \cite{FischerSzabo}, however, the fields are elements of the Gelfand-Shilov spaces $\mathcal{S}_\alpha^\alpha(\R^2)$. Thus, the fields are vanishing rapidly at infinity. If the quadratic potential was absent, it would be clear that this would not be a suitable space for the fields to live in, since it would contain no solution of the free equation of motion. But the quadratic potential does not change this, as can be seen from the absence of poles in the propagator (\ref{eq:PropFischerSzabo}). Thus, the space of fields proposed by \cite{FischerSzabo} does not contain the solutions of the free field equation, even though such solutions exists, as can be seen from the pole in (\ref{eq:Propagator}). To disregard the solutions of the free field equation is certainly a deviation from the principles of perturbative QFT. In particular, it is not clear how to describe asymptotic states (and thus to allow for a contact with experiment).

But also from a mathematical point of view, the approach followed in \cite{FischerSzabo} seems to be questionable, as their basic Theorem 4.2 is incorrect. This can be seen by the following counterexample. In Theorem 4.2 it is stated, that
\[
\lim_{k \to -i \infty} \braket{\chi^k_\pm}{\phi} = 0 \quad \forall \phi \in \mathcal{S}_\alpha^\alpha(\R).
\]
In order to test this assertion, we choose $\phi(q) = e^{-a \lambda^{-2} q^2}$, with some real constant $a$. We have $\chi_k^\pm(q)^* = \eta_k^\pm(q)$ and (the conventions used here are related to those used in \cite{FischerSzabo} by $\mathcal{E} = \lambda^{-2} k/4$, $E' = \lambda^{-2}$, $\nu = - \frac{ik}{4} - \frac{1}{2}$, with the parameters of \cite{FischerSzabo} on the l.h.s.) 
\[
\eta_k^\pm(q) = C i^{-\frac{ik}{8}} \Gamma(\tfrac{ik}{4} + \tfrac{1}{2}) D_{- \frac{ik}{4} - \frac{1}{2}}(\mp \sqrt{2i} \lambda^{-1} q),
\]
where $C$ is some constant independent of $k$ and $D_\nu$ is a parabolic cylinder function. For $\Re \nu<0$, it is given by
\[
D_\nu(z) = \frac{1}{\Gamma(-\nu)} e^{-\frac{1}{4} z^2} \int_0^\infty \ud t \ e^{-zt} e^{-\frac{t^2}{2}} t^{-\nu-1}.
\]
For $\Re \nu < 0$, we compute
\begin{align*}
& \int_{-\infty}^\infty \ud q \ e^{-a \lambda^{-2} q^2} D_\nu(\mp \sqrt{2i} \lambda^{-1} q) \\ & = \frac{\lambda 2^{-\frac{1}{2}}}{\Gamma(-\nu)} \int_{-\infty}^\infty \ud q \ e^{-a\frac{q^2}{2}} e^{- i \frac{q^2}{4}}  \int_0^\infty \ud t \ e^{\pm \sqrt{i} q t} e^{-\frac{t^2}{2}} t^{-\nu-1} \\
& = \frac{\lambda 2^{-\frac{1}{2}}}{\Gamma(-\nu)} \int_0^\infty \ud t \int_{-\infty}^\infty \ud q \ e^{-\frac{2a+i}{4} (q \mp \frac{2 \sqrt{i}t}{2a+i})^2} e^{\frac{it^2}{2a+i}} e^{-\frac{t^2}{2}} t^{-\nu-1} \\
& = \frac{\lambda 2^{-\frac{1}{2}}}{\Gamma(-\nu)} \left( \frac{4\pi}{2a+i} \right)^{\frac{1}{2}} \int_0^\infty \ud t \ e^{-\frac{t^2}{2}( 1 - \frac{2i}{2a+i})} t^{-\nu-1} \\
& = \frac{\lambda}{\Gamma(-\nu)} \left( \frac{\pi}{4a+2i} \right)^{\frac{1}{2}} \int_0^\infty \ud u \ e^{-\frac{u}{2} \frac{2a-i}{2a+i}} u^{-\frac{\nu}{2}-1} \\
& = \frac{\lambda \Gamma(-\frac{\nu}{2})}{\Gamma(-\nu)} \left( \frac{\pi}{4a+2i} \right)^{\frac{1}{2}} \left( \frac{1}{2} \frac{2a-i}{2a+i} \right)^{\frac{\nu}{2}}.
\end{align*}
Here we supposed that $\Re \frac{2a-i}{2a+i} > 0$, i.e., $a>\frac{1}{2}$. It follows that
\[
\braket{\chi^k_\pm}{\phi} = C' i^{-\frac{ik}{8}} \Gamma(\tfrac{ik}{8} + \tfrac{1}{4}) \left( \frac{1}{2} \frac{2a-i}{2a+i} \right)^{-\frac{ik}{8} - \frac{1}{4}}.
\]
It is obvious that this does not converge for $k\to-i\infty$.

\section{The retarded propagator in momentum space}
\label{sec:Fourier}
We compute the Fourier transform
\begin{multline*}
\hat{\Delta}_{\text{ret}}(k_s, l_s, k_t, l_t) \\
= \tfrac{1}{(2\pi)^2} \int \ud u_t \ud v_t \ud u_s \ud v_s \ e^{-i(k_t u_t + l_t v_t + k_s u_s + l_s v_s)} \Delta_{\text{ret}}(u_t, v_t, u_s, v_s)
\end{multline*}
of the retarded propagator. From Remark \ref{rem:EffMass} we know that the retarded propagator can be interpreted as the one for a position dependent mass. The Fourier transform \wrt $k_t$ and $l_t$ is thus well-known, and we obtain
\[
\hat{\Delta}_{\text{ret}}(k_s, l_s, k_t, l_t) = \tfrac{1}{(2\pi)^2} \int \ud u_s \ud v_s \ e^{-i(k_s u_s + l_s v_s)} \frac{-1}{(k_t- i\epsilon) (l_t-i\epsilon) - \frac{u_s v_s}{4 \lambda^4}}.
\]
We now consider the cases $u_s v_s >0$ and $u_s v_s < 0$ seperately. In the first case, we use the coordinates $x=\sqrt{u_s v_s}$, $y = \sqrt{u_s / v_s}$ and obtain the integral
\[
\tfrac{2}{(2\pi)^2} \int_0^\infty \ud x \int_0^\infty \ud y \ \frac{2x}{y} \cos(k_s xy + \tfrac{l_s x}{y} ) \frac{-1}{(k_t- i\epsilon) (l_t-i\epsilon) - \frac{x^2}{4 \lambda^4}}.
\]
In the case $u_s v_s < 0$ we instead find
\[
\tfrac{2}{(2\pi)^2} \int_0^\infty \ud x \int_0^\infty \ud y \ \frac{2x}{y} \cos(k_s xy - \tfrac{l_s x}{y} ) \frac{-1}{(k_t- i\epsilon) (l_t-i\epsilon) + \frac{x^2}{4 \lambda^4}}.
\]
Now for $a, b>0$ one has \cite[(3.868.2 \& 4)]{Gradshteyn}
\begin{align*}
\int_0^\infty \ud x \ \cos(a^2 x + \tfrac{b^2}{x}) \tfrac{1}{x} & = - \pi Y_0(2ab), \\
\int_0^\infty \ud x \ \cos(a^2 x - \tfrac{b^2}{x}) \tfrac{1}{x} & = 2 K_0(2ab).
\end{align*}
In the case $\sign k_s = \sign l_s$ we thus obtain
\begin{multline*}
\tfrac{-4\pi}{(2\pi)^2} \int_0^\infty \ud x \ Y_0(2x\sqrt{\betrag{k_s l_s}}) \frac{-x}{(k_t- i\epsilon) (l_t-i\epsilon) - \frac{x^2}{4 \lambda^4}} \\
+  \tfrac{8}{(2\pi)^2} \int_0^\infty \ud x \ K_0(2x\sqrt{\betrag{k_s l_s}}) \frac{-x}{(k_t- i\epsilon) (l_t-i\epsilon) + \frac{x^2}{4 \lambda^4}},
\end{multline*}
and in the case $\sign k_s \neq \sign l_s$ we get
\begin{multline*}
\tfrac{8}{(2\pi)^2} \int_0^\infty \ud x \ K_0(2x\sqrt{\betrag{k_s l_s}}) \frac{-x}{(k_t- i\epsilon) (l_t-i\epsilon) - \frac{x^2}{4 \lambda^4}} \\
+ \tfrac{-4\pi}{(2\pi)^2} \int_0^\infty \ud x \ Y_0(2x\sqrt{\betrag{k_s l_s}}) \frac{-x}{(k_t- i\epsilon) (l_t-i\epsilon) + \frac{x^2}{4 \lambda^4}}.
\end{multline*}
We have the asymptotic relation \cite[(9.7.2)]{Abramowitz}
\begin{equation*}
K_\nu(z) \sim \sqrt{\frac{\pi}{2z}} e^{-z}
\end{equation*}
and may thus change the contour of integration for the integrals involving $K_0$ to $0 \to - i \infty$. In the first case, we pick up a pole if $k_t+l_t<0$ and in the second one if $k_t+l_t>0$. 
We may then use \cite[(9.6.5)]{Abramowitz}
\[
\pi Y_0(iz) = i \pi I_0(z) - 2 K_0(z)
\]
and $I_0(z) = J_0(iz)$ to obtain
\begin{multline*}
\tfrac{8\pi i}{(2\pi)^2} 4 \lambda^4 H(-k_t-l_t) H(k_t l_t) K_0(-4i\lambda^2 \sqrt{\betrag{k_s l_s k_t l_t}} + \epsilon)  \\
+ \tfrac{8\pi i}{(2\pi)^2} 4 \lambda^4 H(-k_t-l_t) H(-k_t l_t) K_0(4 \lambda^2 \sqrt{\betrag{k_s l_s k_t l_t}} - i \epsilon)  \\
+ \tfrac{-4\pi i}{(2\pi)^2} \int_0^\infty \ud x \ J_0(2x\sqrt{\betrag{k_s l_s}}) \frac{x}{\frac{x^2}{4 \lambda^4} - (k_t- i\epsilon) (l_t-i\epsilon)}
\end{multline*}
for $\sign k_s = \sign l_s$ and
\begin{multline*}
\tfrac{8\pi i}{(2\pi)^2} 4 \lambda^4 H(k_t+l_t) H(k_t l_t) K_0(4 \lambda^2 \sqrt{\betrag{k_s l_s k_t l_t}} - i\epsilon)  \\
+ \tfrac{8\pi i}{(2\pi)^2} 4 \lambda^4 H(k_t+l_t) H(-k_t l_t) K_0(-4i \lambda^2 \sqrt{\betrag{k_s l_s k_t l_t}} + \epsilon)  \\
+ \tfrac{-4\pi i}{(2\pi)^2} \int_0^\infty \ud x \ J_0(2x\sqrt{\betrag{k_s l_s}}) \frac{x}{\frac{x^2}{4 \lambda^4} + (k_t- i\epsilon) (l_t-i\epsilon)}
\end{multline*}
for $\sign k_s \neq \sign l_s$. We have \cite[(6.532.4)]{Gradshteyn}
\[
\int_0^\infty \ud x \ \frac{x}{x^2+k^2} J_0(ax) = K_0(ak), \quad a >0, \Re k >0.
\]
In order to solve the above integrals, we thus have to choose the root with positive real part of $\mp (k_t-i\epsilon)(l_t-i\epsilon)$. For $\sign k_s = \sign l_s$, we obtain
\begin{multline*}
\tfrac{8\pi i}{(2\pi)^2} 4 \lambda^4 H(-k_t-l_t) H(k_t l_t) K_0(-4i\lambda^2 \sqrt{\betrag{k_s l_s k_t l_t}} + \epsilon)  \\
+ \tfrac{8\pi i}{(2\pi)^2} 4 \lambda^4 H(-k_t-l_t) H(-k_t l_t) K_0(4 \lambda^2 \sqrt{\betrag{k_s l_s k_t l_t}} - i \epsilon)  \\
+ \tfrac{-4\pi i}{(2\pi)^2} 4 \lambda^4 H(k_t l_t) K_0(\sign(k_t+l_t) 4i \lambda^2 \sqrt{\betrag{k_s l_s k_t l_t}} + \epsilon)  \\
+ \tfrac{-4\pi i}{(2\pi)^2} 4 \lambda^4 H(-k_t l_t) K_0(4 \lambda^2 \sqrt{\betrag{k_s l_s k_t l_t}} + i \epsilon \sign(k_t+l_t)).
\end{multline*}
This can be written as
\begin{multline*}
\tfrac{-4\pi i}{(2\pi)^2} 4 \lambda^4 \sign(k_t+l_t) \left( H(k_t l_t) K_0(\sign(k_t+l_t) 4 i \lambda^2 \sqrt{\betrag{k_s l_s k_t l_t}} + \epsilon) \right. \\
\left. + H(-k_t l_t) K_0(4 \lambda^2 \sqrt{\betrag{k_s l_s k_t l_t}} + i \epsilon \sign(k_t + l_t)) \right).
\end{multline*}
For $\sign k_s \neq \sign l_s$, one likewise obtains
\begin{multline*}
\tfrac{4\pi i}{(2\pi)^2} 4 \lambda^4 \sign(k_t+l_t) \left( H(k_t l_t) K_0(4 \lambda^2 \sqrt{\betrag{k_s l_s k_t l_t}} - i \epsilon \sign(k_t+l_t) ) \right. \\
\left. + H(-k_t l_t) K_0(-\sign(k_t + l_t) 4 i \lambda^2 \sqrt{\betrag{k_s l_s k_t l_t}} + \epsilon) \right).
\end{multline*}
In total, we thus obtain
\begin{multline*}
\hat{\Delta}_{\text{ret}}(k_s,l_s,k_t,l_t) = - \tfrac{4\pi i}{(2\pi)^2} 4 \lambda^4 \sign(k_t+l_t) \sign(k_s l_s) \\
\times \left( H(k_t l_t k_s l_s) K_0( \sign(k_t+l_t) \sign(k_s l_s) 4 i \lambda^2 \sqrt{\betrag{k_s l_s k_t l_t}} + \epsilon ) \right. \\
\left. + H(-k_t l_t k_s l_s) K_0(4 \lambda^2 \sqrt{\betrag{k_s l_s k_t l_t}} + i \epsilon \sign(k_t+l_t) \sign(k_s l_s)) \right).
\end{multline*}
Note that for large $k_{s/t}, l_{s/t}$, this is bounded, but highly oscillatory in some directions, as expected for an element of ${\mathcal{S}'}^{\alpha, A}(\R^4)$ with $\alpha = \frac{1}{2}$.

\end{document}